\documentclass[twocolumn, 10pt]{svjour3}         
\smartqed  
\usepackage{graphicx}
\usepackage{bm}
\usepackage{amsmath}
\usepackage{amssymb}
\usepackage{latexsym}
\usepackage{epsfig}
\usepackage{amsbsy}
\usepackage{array}
\usepackage{amssymb}
\usepackage{setspace}
\usepackage[caption=false,font=footnotesize]{subfig}
\usepackage{algorithm}
\usepackage{algorithmic}
\usepackage{cite}
\usepackage{setspace}
\DeclareMathOperator*{\E}{\mathbb{E}}

\journalname{Wireless Networks}

\begin{document}

\title{Capacity and Delay-Throughput Tradeoff in ICMNs with Poisson Meeting Process}


\author{Yin~Chen \and Yulong~Shen \and Jinxiao~Zhu \and Xiaohong~Jiang}

\institute{Y.~Chen, J.~Zhu and X.~Jiang \at School of Systems Information Science, Future University Hakodate, Japan.\\
							\email{ychen1986@gmail.com, jxzhu1986@gmail.com and jiang@fun.ac.jp}.   
           \and
           Y.~Shen \at School of Computer Science and Technology, Xidian University,  China.\\
           \email{ylshen@mail.xidian.edu.cn}. \\
}



\date{Received: date / Accepted: date}

\maketitle

\begin{abstract}

Intermittently connected mobile networks \\(ICMNs) serve as an important network model for many critical applications. This paper focuses on a continuous ICMN model where the pair-wise meeting process between network nodes
follows a homogeneous and independent Poisson process. This ICMN model is known to serve as good approximations to a class of important ICMNs with mobility models like random waypoint and random direction, so it is widely adopted in
the performance study of ICMNs. This paper studies the throughput capacity and delay-throughput tradeoff  in the considered ICMNs with Poisson meeting process.
For the concerned ICMN, we first derive an exact expression of its throughput capacity based on the pairwise meeting rate therein and analyze the expected end-to-end packet delay under a routing algorithm that can achieve the throughput capacity. We then explore the inherent tradeoff between delay and throughput and establish a necessary condition for such tradeoff that holds under any routing algorithm in the ICMN. To illustrate the applicability of the theoretical results, case studies are further conducted for the random waypoint and random direction mobility models. 
Finally, simulation and numerical results are provided to verify the efficiency of our theoretical capacity/delay results and to illustrate
our findings.

\keywords{Intermittently Connected Mobile Network (ICMN) \and Delay Tolerant Networks (DTN)  \and Throughput Capacity \and  Delay-Throughput tradeoff \and Poisson Process}
\end{abstract}

%



\section{Introduction}\label{sec:intro}

Intermittently connected mobile networks (ICMNs) or delay tolerant networks (DTNs) represent a class of sparse mobile ad hoc networks (MANETs), where a collection of self-autonomous mobile nodes communicate with each other via peer-to-peer wireless links without any support from preexisting infrastructures, but complete end-to-end path(s) between a node-pair may never exist so nodes mainly rely on mobility as well as basic packet storing, carrying, and forwarding operations to implement end-to-end communication (see e.g.,~\cite{Zhang2006} for a survey). ICMNs are highly flexible, robust and  rapidly deployable and reconfigurable, so they serve as an important model for many critical applications such as wildlife tracking and monitoring, battlefield communication, vehicular networks, low-cost Internet service for remote communities.

By now, much academic activity has been devoted to the performance study on ICMNs.
In the seminal work of~\cite{Groenevelt2005, Groenevelt2005a}, Groenevelt et al. demonstrated that the ICMN model with Poisson meeting process can approximately fit an important class of mobility models such as random waypoint, random direction and   random walk. 
Based on this ICMN model, the authors of~\cite{Groenevelt2005a} conducted  Markov chain-based analysis to evaluate the  performance under  two-hop routing and epidemic routing algorithms in terms of the packet delivery delay, i.e., the time it takes for a packet to reach its destination node after it departures from its source node.
Following this work, the packet delivery delay performance was extensively studied in literature~\cite{Zhang2007,AlHanbali2008,Spyropoulos2008,Spyropoulos2008a}.
Notice that while the Markov chain-based analysis enables the distribution of delivery delay to be calculated, the analysis quickly becomes cumbersome and computationally impractical as the network size (i.e., the number of network nodes) increases.
Motivated by this observation, Zhang et al.~\cite{Zhang2007} developed a theoretical framework based on ordinary differential equations which significantly reduce the complexity involved in the delivery delay analysis for large scale ICMNs.
For ICMNs with two-hop routing and packet life time constraint and ICMNs with spray and wait routing, the corresponding delivery delay performance was reported in~\cite{AlHanbali2008} and~\cite{Spyropoulos2008,Spyropoulos2008a}, respectively.
For the throughput performance, Subramanian  et al. explored the achievable throughput of ICMNs under two-hop routing~\cite{Subramanian2009Proc.IEEEISIT,Subramanian2009Proc.IEEEWiOPT} as well as under multi-hop routing~\cite{Subramanian2012Proc.IEEEICC}.

While the above works are helpful for us to have a preliminary understanding on the performance of ICMNs, further deliberate studies are needed to reveal the fundamental performance limits of such networks. 
First, the available throughput studies discussed above~\cite{Subramanian2009Proc.IEEEISIT,Subramanian2009Proc.IEEEWiOPT,Subramanian2012Proc.IEEEICC} only focus on the  throughput study in ICMNs under a specified routing algorithm, the  throughput capacity, i.e., the maximum  throughput over any routing algorithm, is still unknown for the ICMN model with Poisson meeting process.
Second, the studies on delivery delay, which constitutes only a part of the fundamental end-to-end packet delay, can not be directly applied to investigate the inherent tradeoff between the end-to-end delay and  throughput in ICMNs.
Since the throughput capacity and  delay-throughput tradeoff in ICMNs indicate the ``best'' performance (i.e., theoretical  limits) that the network can stably support, it is expected that understanding these fundamental performance limits will provide profound insight to facilitate the design and optimization for these networks~\cite{Goldsmith2011}.

In this paper, we focus on the ICMNs with Poisson meeting process and study the throughput capacity and inherent delay-throughput tradeoff in such networks, where the proof techniques are inspired by the prior work of Neely and Modiano in~\cite{Neely2005}. 
The main difference between~\cite{Neely2005} and this work is the  network models under study.
The work of~\cite{Neely2005} focused on a time-slotted and cell-partitioned network model where the network nodes there move following an i.i.d. mobility model.
We study in this paper a time continuous ICMN model with Poisson meeting process, which is known to serve as a good approximation to a more general and important class of mobility models~\cite{Groenevelt2005, Groenevelt2005a} and hence has been widely adopted in the performance study for ICMNs.
%
The main contributions of the paper are summarized as follows.

\begin{itemize}
	\item For the concerned ICMN model with Poisson meeting process, we first derive an exact expression on its throughput capacity based on the pairwise meeting rate between network nodes there. 
	The analysis on the expected end-to-end packet delay under one capacity achieving routing algorithm is also provided. 	
	\item We then explore the inherent tradeoff between the expected end-to-end packet delay and  throughput  and establish a necessary condition for such tradeoff that holds under any routing algorithm in the concerned ICMNs.
	\item Case studies for typical random waypoint and random direction mobility model are further conducted to illustrate the applicability of our  theoretical results on the throughput capacity and delay-throughput tradeoff developed in this paper.
	\item Finally, we provide simulation/numerical results to  verify the efficiency of our theoretical capacity/delay results and to illustrate our  findings.
\end{itemize}

The rest of the paper is outlined as follows.
The related work is introduced in Section~\ref{sec:related_work}.
Section~\ref{sec:model} presents system models and some basic definitions.
The main theoretical results on throughput capacity and delay-throughput tradeoff are derived in Section~\ref{sec:capacity}.
Section~\ref{sec:numerical} provides simulation/numerical results and corresponding discussions.
Finally, we conclude this paper in Section~\ref{sec:conclusion}.

%
%
%
%

\section{Related Works}\label{sec:related_work}



Since the seminal work of Grossglauser and Tse~\cite{Grossglauser2002}, the throughput capacity and  delay-throughput tradeoff have been extensively studied for MANETs under various mobility models, most of which focused on deriving  order-sense results and scaling laws, i.e., to find asymptotic bounds $\Theta(f(n))$ for throughput capacity as a function of number of network nodes $n$\footnote{In this paper, for two functions $f(n)$ and $g(n)$, we denote $f(n)=O(g(n))$ iff there exist positive constants $c$ and $n_0$, such that for all $n \geq n_0$, the inequality $0 \leq f(n) \leq c g(n)$ is satisfied; $f(n) = \Omega(g(n))$  iff $g(n) = O(f(n))$; $f(n)=\Theta(g(n))$ iff both $f(n)=O(g(n))$ and $f(n)=\Omega(g(n))$ are satisfied.}.
The result of~\cite{Grossglauser2002} indicates that the long-term per flow throughput can be kept constant even as $n$ tends to infinity.
Gamal et al.~\cite{Gamal2006a, Gamal2006} studied a cell-partitioned MANET divided evenly into $n \times n$ cells, on which the nodes move independently according to a symmetric random walk.
For the considered MANET, the authors of~\cite{Gamal2006a, Gamal2006} investigated its optimal scaling behavior of the delay-throughput tradeoff and discovered that the $\Theta (1)$ per flow throughput is achievable  at the cost of an average delay of order $\Theta(n \log{n})$.
A similar delay-throughput tradeoff was shown to also exist in MANETs under restricted mobility model~\cite{Mammen2007}.
In the work of~\cite{Li2012a}, Li et al. proposed a controllable mobility model for cell-partitioned MANETs and 
derived upper and lower bounds on  the achievable throughput and expected  delay for the considered networks.
Besides, the scaling laws of the throughput capacity and related delay-throughput tradeoff have also been explored under other mobility models, such as  Brownian mobility model~\cite{Gamal2004Proc.IEEEINFOCOM,Lin2006},   hybrid mobility model~\cite{Sharma2007}, correlated mobility model~\cite{ciullo2011impact} and ballistic mobility model~\cite{Bogo2013Proc.ACMMOBICOM}. 
For a survey on the scaling law results of throughput capacity and delay in wireless networks, please refer to~\cite{Lu2013}.

It is notable that although the study on order sense results and scaling laws can help us to understand the asymptotic behavior of the throughput capacity and delay-throughput tradeoff as the number of network nodes increases, they provide little information on the actually achievable throughput/delay performance of these networks, which is of more interest from the view of network designers.
Noting the limitation of scaling law results, some preliminary work has been conducted for the exact expressions of throughput capacity of MANETs~\cite{Neely2005,Urgaonkar2011,JuntaoGao2013,Chen2013Proc.IEEEICCC}.
In particular, Neely and Modiano~\cite{Neely2005} computed the exact throughput capacity and delay-throughput tradeoff in a cell-partitioned MANET under an i.i.d. mobility model where the  locations of each network node in steady-state are independently and uniformly distributed over all cells.
Following the model of~\cite{Neely2005}, Urgaonkar and Neely further investigated the relation between throughput capacity and energy consumption in~\cite{Urgaonkar2011}.
Recently, Chen et. al~\cite{Chen2013Proc.IEEEICCC} studied the exact throughput capacity for a continuous MANET with the i.i.d. mobility model and an ALOHA  protocol for medium access control.

Despite the insight provided by existing exact results on the throughput capacity, the results developed there largely rely on an independent and uniform distribution of the locations of network nodes  in steady-state and hence are only applicable to networks under the i.i.d. mobility model.
This paper studies the exact throughput capacity and related delay-throughput tradeoff under a more widely accepted ICMN model and the result developed in this analysis can be applied to ICMNs under a general class of mobility models that can approximately fit the Poisson meeting process, irrespective of the stationary distribution of the locations of network nodes.

\section{System Models and Definitions}\label{sec:model}
In this section, we first introduce the network model, mobility model and traffic model, and then define the performance metrics involved in this study.
\subsection{Network Model}

We consider a  sparse network that consists of $n$ identical mobile nodes randomly moving within a continuous square of side-length $L$.
Each node has a maximum transmission distance $d$.
We call that two nodes ``meet'' when their distance is less than $d$ and thus they can conduct communication.
At the beginning of each meeting, either of the two nodes is randomly selected  as the transmitter of this meeting with equal probability.
Since the network is very sparse, we assume that the effect of interference is negligible.
The total number of bits transmitted during a meeting is fixed and normalized to one packet.

\subsection{Mobility Model}

We consider a general  model introduced in~\cite{Groenevelt2005} for node mobility.
Under this mobility model, the meeting process between each pair of nodes can be modeled as mutually independent and  homogeneous Poisson processes with rate $\beta > 0$.
Equivalently stated, the pairwise inter-meeting times, i.e., the time that elapses between two consecutive meetings of a given pair of nodes,  are mutually independent and exponentially distributed with  mean $1/\beta$.
It has been demonstrated in previous studies that this mobility model serves as good approximations to  a lot of typical mobility models like random waypoint, random direction and random walk models~\cite{Groenevelt2005,AlHanbali2008,Zhang2007}.
Specifically, the result of~\cite{Groenevelt2005} shows that for ICMNs with the  random waypoint  (RW) and the random direction (RD) models, the corresponding pairwise meeting rates $\beta_\text{RW}$ and $\beta_\text{RD}$ can be efficiently approximated as
%
%
\begin{equation}\label{eqn:beta_estimate}
	\beta_{\text{RW}} \approx \frac{ 2 c_1 \, d \E[V^*]}{L^2},~~\mbox{and}~~
	\beta_\text{RD} \approx \frac{2 d \E[V^*]}{L^2},
\end{equation}
respectively, where $c_1 = 1.3683$ is a constant and  $\E[V^*]$ is the average relative speed between two nodes (see~\cite{Groenevelt2005a} for the numerical calculation of $\E[V^*]$). In the special case that each node travels at a constant speed  $v$, we have $\beta_{\text{RW}} \approx \frac{8 c_1 d v}{\pi L^2}$ and $\beta_\text{RD} \approx \frac{8 d v}{L^2} $.


\subsection{Traffic Model}

Regarding traffic pattern, we consider the permutation traffic model~\cite{ciullo2011impact}.
Under this model, there are $n$ unicast traffic flows in the network and each node is the source of one traffic flow and also the destination of another traffic flow. 
Let $\varphi(i) \neq i$ denote the destination node of the traffic flow originated from node $i$, $i = 1,2,\ldots  n$, the source-destination pairs are matched at random in the sense that the sequence $(\varphi(1), \varphi(2),\ldots \varphi(n))$ is just a  permutation of the set of nodes $\{1, 2, \ldots n\}$.
The packet arrival process at each node is assumed to be a Poisson arrival process  with rate $\lambda > 0$.
For throughput capacity analysis, we consider that there is no constraint on packet life time and the buffer size in each node is sufficiently large such that packet loss due to buffer overflow will never happen.

\subsection{Performance Metrics}\label{subsec:definition}
The performance metrics involved in this study are defined as follows.

\textbf{End-to-end packet delay}: The \emph{end-to-end delay} of a packet is the time it takes for the packet to reach its destination after it arrives at its source.
 
\textbf{Network stability}: For an ICMN under a routing algorithm, if the packet arrival rate  to each node is $\lambda$, the network is called \emph{stable} under this rate if the average number of packets waiting at each node, i.e., the average queue length, does not grow to infinity with time and thus the average end-to-end packet delay is bounded.

\textbf{Throughput capacity}: The \emph{throughput capacity} $\mu$ of the concerned ICMN is defined as the maximum value of packet arrival rate $ \lambda$  that the network can stably support over any possible routing algorithm.


\section{Throughput Capacity and Delay-Throughput Tradeoff}\label{sec:capacity}
In this section, we first establish a theorem regarding the throughput capacity result in the considered ICMN based on the pairwise meeting rate therein, and provide necessity and sufficiency proofs for this theorem.
Then, we proceed to explore the tradeoff between the end-to-end delay  and throughput.
Finally,  specific case studies are further conducted  for ICMNs under the random waypoint and random direction mobility models.

\subsection{Throughput Capacity}
\begin{theorem}\label{theorem:capacity}
For the concerned ICMN with $n$ mobile nodes and pairwise meeting rate $\beta$, its throughput capacity can be determined as 
\begin{equation}\label{eqn:capacity}
	\mu = \frac{n}{4} \beta.
\end{equation}
\end{theorem}

The proof of  Theorem~\ref{theorem:capacity} involves  proving that  $\lambda \leq \mu$ is necessary and $\lambda < \mu$ is sufficient to  ensure network stability.
We establish the necessity in Section~\ref{sec:necessity} by showing that $\mu$ is an upper bound on the  throughput under any possible routing algorithm in the considered ICMN.
Then, we prove the sufficiency in Section~\ref{sec:sufficiency}, where  a routing algorithm is presented and it is shown that the network is stable under this routing algorithm for any  rate $\lambda < \mu$.
The proof of Theorem~\ref{theorem:capacity} follows the techniques developed in~\cite{Neely2005}.

\subsubsection{Proof of Necessity}\label{sec:necessity}

\begin{lemma}\label{lemma:necessity}
For the concerned ICMN with $n$ mobile nodes and pairwise meeting rate $\beta$, its throughput under any possible routing  algorithm  is upper bounded by
	\begin{equation}
	\mu = \frac{n}{4} \beta.
\end{equation}
\end{lemma}
\begin{proof}
Consider any possible routing algorithm.
Let  $X_h(T)$ denote the total number of packets transferred through $h$ hops from their sources to  destinations in time interval $[0,T]$.
Notice that to ensure network stability, the sum of arrival rates of all traffic flows should be not greater than the sum of throughputs, since otherwise the amount of packets waiting in the network will grow to infinity as time evolves.
Formally, it is necessary that for any given $\epsilon > 0$, there must exist an arbitrarily large $T$ such that the following inequality  holds
\begin{equation}\label{eqn:stability}
	\lambda n - \epsilon \leq \frac{1}{T} \sum_{h=1}^{\infty} X_h(T),
\end{equation}
where $\lambda$ denotes the packet arrival rate at each node.

Notice the fact that during the time interval $[0,T]$, the total number of packet transmissions is lower bounded by $\sum_{h=1}^{\infty} h X_h(T)$ and upper bounded by the total number of meetings  between all node pairs during this time interval, denoted by $Y(T)$ in the following.
Thus,  we have from the transitivity that
\begin{equation}\label{eqn:tran_upp_bound}
	\sum_{h=1}^{\infty} h X_h(T) \leq Y(T).
\end{equation}
From~(\ref{eqn:stability}) and~(\ref{eqn:tran_upp_bound}), we have
\begin{align}
\frac{1}{T} Y(T) &\geq  \frac{1}{T} X_1 (T) +  \frac{2}{T} \sum_{h=2}^{\infty} X_h(T)\nonumber\\%
								 &\geq  \frac{1}{T} X_1 (T) + 2\left[(\lambda n - \epsilon) - \frac{1}{T} X_1 (T) \right],
\end{align}
and thus 
\begin{equation}\label{eqn:bound}
\lambda \leq \frac{1}{2n}\left [\frac{1}{T} Y(T) + \frac{1}{T} X_1 (T) + 2\epsilon \right].
\end{equation}

Since a packet can be transferred from its source to destination through single hop only when the source conducts a  transmission directly to the destination, the term $X_1 (T)$ in~(\ref{eqn:bound}), i.e., the number of packets transferred from source to destination within one hop during $[0,T]$,  is upper bounded by $Y_{sd}(T)$, i.e., the number of direct transmissions  from each source  node to its destination during the time interval $[0,T]$.
Notice that in the network there are $\binom{n}{2} = \frac{(n-1) n}{2}$ node-pairs and based on the property of the Poisson meeting process, the meeting rate of each pair of nodes is $\beta$.
It follows that the expectation of the number of transmissions occurring in the network is just equal to $\frac{(n-1) n}{2} \beta$.
Applying the law of large numbers, we have as $T \to \infty$
\begin{equation}
\frac{1}{T}Y(T)  \xrightarrow{\text{a.s.}} \frac{(n-1) n}{2} \beta   \label{eqn:ex_no_tr}.
\end{equation}
Similarly, the expectation of the number of transmissions conducted  from source nodes to their destination directly is equal to $\frac{n}{2}  \beta$, so as $T \to \infty$
\begin{equation}
\frac{1}{T}Y_{sd}(T)  \xrightarrow{\text{a.s.}} \frac{n}{2}  \beta \label{eqn:ex_no_sdtr}.
\end{equation}
Using~(\ref{eqn:ex_no_tr}) and~(\ref{eqn:ex_no_sdtr}) into~(\ref{eqn:bound}), it follows that

\begin{algorithm}[!t]
\caption{Routing Algorithm.}
\label{alg:routing}
\begin{algorithmic}[1]
\STATE Suppose that there is a meeting between two nodes, transmitter  \emph{Tx} and  receiver \emph{Rx}, respectively.
\IF {\emph{{Rx}} is the destination of the traffic generated from \emph{{Tx}}}
\STATE {\emph{{Tx}} conducts a \emph{source-to-destination} transmission:}
\IF{\emph{Tx} has packet(s) in its local queue}
\STATE {\emph{Tx} transmits the head-of-line packet of the queue
to \emph{Rx}.}
\ELSE
\STATE{\emph{Tx} remains idle.}
\ENDIF
\ELSE
\STATE {\emph{Tx} flips an unbiased coin;}
\IF{it is the head}
\STATE {\emph{Tx} conducts a \emph{source-to-relay} transmission:}
\IF{\emph{Tx} has packet(s) in its local queue }
\STATE {\emph{Tx} transmits the head-of-line packet of the queue to \emph{Rx}.}
\ELSE
\STATE{\emph{Tx} remains idle.}
\ENDIF
\ELSE
\STATE {\emph{Tx} conducts a \emph{relay-to-destination} transmission:}
\IF{\emph{Tx} has packet(s) in the relay queue destined for \emph{Rx}}
\STATE {\emph{Tx}  the head-of-line packet of the queue to \emph{Rx}.}
\ELSE
\STATE{\emph{Tx} remains idle.}
\ENDIF
\ENDIF
\ENDIF
\end{algorithmic}
\end{algorithm}
\begin{equation}
\lambda \leq \frac{n}{4} \beta + \frac{\epsilon}{n}, \text{  as } T \to \infty.
\end{equation}
Since  $\epsilon$ can be arbitrarily small, the result then follows.
\end{proof}

\subsubsection{Proof of Sufficiency}\label{sec:sufficiency}
For the proof of sufficiency, we  present a routing algorithm in Algorithm~\ref{alg:routing} and  derive the expected end-to-end packet delay in the considered ICMN under this routing algorithm in Lemma~\ref{lemma:sufficiency}.
To support the operation of Algorithm~\ref{alg:routing}, we assume that each node maintains one source queue to store packets locally generated  and  $n - 2$  relay queues to store packets of other flows (one queue per flow). All these queues follow the FIFO (first-in-first-out) discipline. The proof of Lemma~\ref{lemma:sufficiency} uses the reversibility of continuous time $M/M/1$ queues.


\begin{lemma}\label{lemma:sufficiency}
For the concerned ICMN with $n$ mobile nodes and  pairwise meeting rate $\beta$, if the packet arrival process at each node is an i.i.d. Poisson  process with rate $\lambda$ and Algorithm~\ref{alg:routing} is adopted for packet routing,   the corresponding expected end-to-end delay $\E\{D\}$ is determined as
\begin{equation}\label{eqn:delay}
	\E\{D\} = \frac{n-1}{\mu-\lambda},
\end{equation}
where $\mu$ is the upper bound determined in Lemma~\ref{lemma:necessity}, 
\end{lemma}

\begin{proof}
Notice that under Algorithm~\ref{alg:routing}, there are three types of transmissions, i.e., source-to-destination transmission, source-to-relay transmission and relay-to-destination transmission.
It takes a packet at most two hops to reach its destination and the packet delivery processes of the $n$ traffic flows are independent from each other.
Based on the properties of the mobility model and Algorithm~\ref{alg:routing}, we can see that the packet delivery process in the considered ICMN under Algorithm~\ref{alg:routing} consists of $n$ identical queuing processes (one queuing process per flow). 
Without loss of generality, we focus on in the analysis the queuing process of an arbitrary traffic flow illustrated in Fig.~\ref{fig:routing_process}.
It can be seen from Fig.~\ref{fig:routing_process} that packets  of this flow experience a  two-stage queuing process if the packet is not directly transmitted to the destination, i.e., the queuing process at the source node (first stage) and the queuing process at one of the $n-2$ relay nodes (second stage).

\begin{figure}[!t]
	\centering
		\includegraphics[width=3.0in]{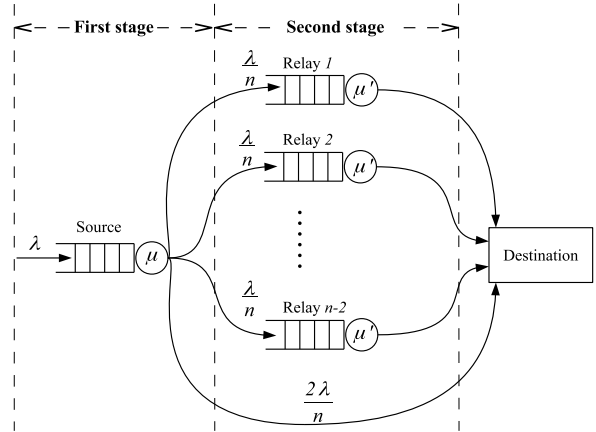}
	\caption{Two-stage queuing process under Algorithm~\ref{alg:routing}. In the figure, the inter-service times in the source node and relay nodes are exponentially distributed with rate $\mu = \frac{n}{4}\beta$ and rate $\mu' = \frac{\beta}{4}$, respectively.}
	\label{fig:routing_process}
\end{figure}

Consider first the source queue.
The input to this queue is a Poisson arrival process with rate $\lambda$.
According to Algorithm~\ref{alg:routing}, a ``service'' comes  when  the source node conducts either a \emph{source-to-destination} transmission or a  \emph{source-to-relay} transmission.
Based on the property of Poisson meeting process and Algorithm~\ref{alg:routing}, the service process is a Poisson process with service rate equal to
\begin{align}
\mu &= \beta /2 + \beta (n-2) /4 \label{eqn:service_rate}\\
	&= \frac{n}{4}\beta,
\end{align}
where the first term in~(\ref{eqn:service_rate}) is the rate associated with the particular source meeting its destination and multiplied by $1/2$ for the probability that the source is chosen to transmit,  and the second term is the rate of this source meeting any one of the $n-2$ relay nodes and  multiplied by the $1/4$ for the probability that the source is chosen to transmit and the source-to-relay transmission is selected.
%
%
Then, it follows that the source queue is an $M/M/1$ queue with input rate $\lambda$ and service rate $\mu$.
Based on the result from queuing theory, the mean queuing delay of the source queue  $\E\{D_s\}$ is given by
\begin{equation}
	\E\{D_s\} = \frac{1}{\mu-\lambda}.
\end{equation}
Moreover, since $M/M/1$ queues are reversible, so the departure process from the source queue is also a Poisson process with rate $\lambda$~\cite{kelly2011reversibility}.

Consider now the queuing process at one of the $n-2$ relay nodes.
Notice that with probability $\frac{1}{n}$ a packet departure from the source node will enter this relay node, so the input to this relay queue is a Poisson process with rate $\frac{\lambda}{n}$.
In this relay queue, a ``service'' arises when this relay node conducts a \emph{relay-to-destination} transmission to the destination node of the concerned traffic flow, so the service process of the relay nodes is a Poisson process with rate $\mu' = \frac{\beta}{4}$.
We can see that the relay queue is again an $M/M/1$ queue.
The mean queuing delay $\E\{D_r\}$ at a relay node is given by
\begin{equation}
	\E\{D_r\} = \frac{1}{\mu'-\lambda/n}.
\end{equation}

Summing up the above results, we have that the expected end-to-end packet delay is 
\begin{equation}
	\E\{D\} = \E\{D_s\} + \frac{n-2}{n} \E\{D_r\} = \frac{n-1}{\mu-\lambda},
\end{equation}
which proves the lemma.
\end{proof}


\subsection{Delay-Throughput Tradeoff}\label{sec:tradeoff}
In the following theorem, we establish a necessary condition on the tradeoff between the end-to-end packet delay and achievable throughput  under any routing algorithm  that stabilizes the network. The proof follows the  technique developed in~\cite{Neely2005}.

\begin{theorem}\label{theorem:tradeoff}
Consider an ICMN with $n$ mobile nodes and pairwise meeting rate $\beta$ and  the packet arrival rate at each node is $\lambda$.
A necessary condition for any routing algorithm that can stabilize the network with rate $\lambda$ while maintaining a bounded expected end-to-end delay $\E\{D\}$ is given by
\begin{equation}
	\frac{\E\{D\}}{\lambda} \geq \frac{1- \log(2) }{2 (n-1) \beta^2}.
\end{equation}
\end{theorem}
\begin{proof}
Consider that the packet arrival rate to each of the $n$ traffic flows is $\lambda$ and that there is a general routing algorithm that stabilizes the network under this rate and results in an expected end-to-end delay of $\E\{D\}$.

Let random variable $D_i$ denote the end-to-end delay of a packet in flow $i$ under the routing algorithm and $\E\{D_i\}$ represent its expectation, then the expected end-to-end packet delay of the  network $\E\{D\}$ can be calculated by
\begin{equation}\label{eqn:total_delay_ICN}
	\E\{D\} = \frac{1}{n} \sum_{i=1}^{n} \E\{D_i\}.
\end{equation}
Let random variable $R_i$ denote the redundancy of a packet in flow $i$, i.e., this packet is distributed into $R_i$ different nodes (including the destination) in the network, and $\E\{R_i\}$ be its expectation.
Notice that the sum of the generating rates of  packet redundancy in the network  is
\begin{equation}
 \lambda n \cdot \frac{1}{n} \sum_{i=1}^{n} \E\{R_i\}	= \lambda \sum_{i}^{n} \E\{R_i\}.
\end{equation} 
This quantity is upper bounded by the sum of pairwise meeting rates  in the network, due to the fact  that during each meeting  at most one copy of a packet is transmitted from one node to another.
Formally, it is expressed as 
\begin{equation}\label{eqn:rate_redun_ICN}
	\lambda \sum_{i=1}^{n} \E\{R_i\} \leq \binom{n}{2} \beta = \frac{(n-1)n}{2} \beta.
\end{equation}

For traffic flow $i$, its expected end-to-end delay $\E\{D_i\}$ satisfies the following inequality
\begin{align}
	\E\{D_i\} &= \E  \left\{ D_i | R_i \leq 2 \E  \left\{ R_i \right \} \right\} \Pr \left \{R_i \leq 2 \E\{R_i\}   \right \} \nonumber\\
	&~~~+\E \{D_i | R_i > 2 \E\{R_i\} \} \Pr \{R_i > 2 \E\{R_i\}  \} \nonumber\\
	&\geq \E \{D_i | R_i \leq 2 \E\{R_i\} \} \Pr \{R_i \leq 2 \E\{R_i\}  \} \nonumber\\
	&\geq  \frac{1}{2} \E \{D_i | R_i \leq 2 \E\{R_i\} \},\label{eqn:one_half}
\end{align}
where~(\ref{eqn:one_half}) is due to that $\Pr \{R_i \leq 2 \E\{R_i\}  \} \geq  \frac{1}{2}$ holds for any non-negative random variable.
Now, we consider a virtual network where there are $n$ nodes and $ 2 \E\{R_i\} $ of them initially possess a copy of a packet destined for some other node. 
Let  $D_i^*$ denote the time elapsed from the initial moment until the  moment that one of the $ 2 \E\{R_i\} $ nodes  meets the destination node of the packet,
then  $D_i^*$ is exponentially distributed with parameter $2 \E\{R_i\} \beta$, so that  $\E\{D_i^*\} = \frac{1}{2 \E\{R_i\} \beta}$.

Notice that $\E \{D_i | R_i \leq 2 \E\{R_i\} \}$ is not necessarily lower bounded by $\E\{D_i^*\}$, because the redundancy  $R_i$ may be correlated with certain events in the mobility process, so conditioning on the event $\{R_i \leq 2 \E\{R_i\} \}$ may skew the memoryless property of the Poisson meeting process. However, since $\Pr \{R_i \leq 2 \E\{R_i\}  \} \geq  \frac{1}{2}$, we have the following bound:
\begin{equation}\label{eqn:mini_ineq}
	\E \{D_i | R_i \leq 2 \E\{R_i\} \} \geq \underset{\Theta}{\inf}  \E\{D_i^* | \Theta \},
\end{equation}
where the left-side conditional expectation is minimized over all possible events $\Theta$ that occurs with probability greater than or equal to $1/2$. 
The inequality holds because the event yielding the mobility patterns of the type encountered when $\{R_i \leq 2 \E\{R_i\} \}$ is also included in the events set, over which the conditional expectation is minimized.

Notice that since $D_i^*$ is a continuous variable, so the event minimizing the conditional expectation in~(\ref{eqn:mini_ineq}) is just  $\{ D_i^* \leq \omega \}$ such that $\omega$ is the  smallest value satisfying  $\Pr\{ D_i^* \leq \omega \} = \frac{1}{2}$.
Since $D_i^*$  is exponentially distributed with rate $2 \E\{R_i\} \beta$, so  $\omega = \frac{\log (2)}{2 \E\{R_i\} \beta}$ and $\underset{\Theta}{\inf}  \E\{D_i^* | \Theta \}$ is determined as 
\begin{align}
	\underset{\Theta}{\inf}  \E\{D_i^* | \Theta \} &= \E\{D_i^* | D_i^* \leq \omega \}\nonumber\\
	&= \frac{\E\{D_i^*\} -\E\{D_i^* | D_i^* > \omega \}\Pr\{ D_i^* > \omega \} }{\Pr\{ D_i^* \leq \omega \}}\nonumber\\
	&=\frac{\frac{1}{2 \E\{R_i\} \beta} - \frac{1}{2}(\omega + \frac{1}{2 \E\{R_i\} \beta})}{1/2}\nonumber\\
	&= \frac{1-\log (2)}{{2 \E\{R_i\} \beta}}.\label{eqn:mini_R_star}
\end{align}

Substituting~(\ref{eqn:mini_R_star}),~(\ref{eqn:mini_ineq}) and~(\ref{eqn:one_half})  into~(\ref{eqn:total_delay_ICN}) leads to
\begin{align}
	\E\{D\} 
	&\geq \frac{1- \log(2)}{4 \beta} \cdot \frac{1}{n} \sum_{i=1}^{n}  \frac{1}{\E\{ R_i \}}\label{eqn:jen1_ICN}\\
	&\geq  \frac{1- \log(2)}{4 \beta} \cdot \frac{1}{\frac{1}{n} \sum_{i=1}^{n}   {\E\{ R_i \}}},\label{eqn:jen2_ICN}
\end{align}
where~(\ref{eqn:jen2_ICN}) results  from Jensen's inequality, since the function $f(x) = 1/x$ is convex for $x>0$.
Combining ~(\ref{eqn:rate_redun_ICN}) and~(\ref{eqn:jen2_ICN}), we have
\begin{align}
	\E\{D\} &\geq    \frac{1- \log(2)}{4 \beta} \cdot \frac{2 \lambda}{(n-1)\beta} =\frac{1- \log(2) }{2 (n-1) \beta^2} \cdot \lambda .\label{eqn:delay_bound_ICN}
\end{align}
Multiplying $1/\lambda$ on both sides of (\ref{eqn:delay_bound_ICN}) proves the theorem.
\end{proof}

\subsection{Case Studies under Random Waypoint and Random Direction Models}

So far, we have derived the throughput capacity and delay-throughput tradeoff for the concerned ICMNs with Poisson meeting process.
To illustrate the applicability of these theoretical results, we also do case studies for the random waypoint and random direction mobility models, where parameter-matching is conducted on these model to fit the studied Poisson meeting process. It will be demonstrated in Section~\ref{sec:numerical} via simulation that the  results derived here can serve as good approximations for networks under these mobility models.



\emph{Throughput Capacity:}
For an ICMN with $n$ mobile nodes, side-length $L$ and maximum transmission distance $d$,  when $d \ll L$, 
the throughput capacities  $\mu_\text{RW}$ under the  random waypoint  model and $\mu_\text{RD}$ under the random direction   model  can be efficiently approximated as
\begin{equation}\label{eqn:appr_capa}
	\mu_{\text{RW}} \approx \frac{ c_1 n d   \E[V^*]}{ 2 L^2}~~\mbox{and},~~
	\mu_\text{RD} \approx \frac{ n d   \E[V^*]}{2 L^2},
\end{equation}
respectively, where $c_1 = 1.3683$ is a constant and  $\E[V^*]$ is the average relative speed between a pair of nodes. In the special case of constant traveling speed  $v$, we have $\mu_{\text{RW}} \approx \frac{2 c_1 n d v}{ \pi L^2}$ and $\mu_\text{RD} \approx \frac{2 n d  v}{L^2} $, respectively.

\emph{Delay-throughput tradeoff:}
For an ICMN with $n$ mobile nodes, side-length $L$ and maximum transmission distance $d$,
when $d \ll L$, a necessary condition for any routing algorithm that can stabilize the network with packet arrival rate $\lambda$ while maintaining a bounded expected end-to-end delay $\E\{D\}$ is given by
\begin{enumerate}
	\item for the random waypoint mobility model:
\begin{equation}
	\frac{\E\{D\}}{\lambda} \geq \frac{  (1-\log (2) )L^4}{8(n-1) (c_1 d  \E [V^*])^2   },
\end{equation}
  \item for the random direction mobility model:
\begin{equation}
	\frac{\E\{D\}}{\lambda} \geq \frac{   (1-\log (2) ) L^4}{8(n-1) (  d  \E [V^*])^2   },  
\end{equation} 
\end{enumerate}
where $c_1 = 1.3683$ is a constant and  $\E[V^*]$ is the average relative speed between a pair of nodes. In the special case of constant traveling speed  $v$, the necessary condition is given by

\begin{enumerate}
	\item for the random waypoint mobility model:
\begin{equation}
	\frac{\E\{D\}}{\lambda} \geq \frac{  (1-\log (2) ) \pi^2 L^4}{128 (n-1) (c_1 d  v)^2   },\label{eqn:tradeoff_rw}
\end{equation}
 \item for the random direction mobility model:
\begin{equation}
	\frac{\E\{D\}}{\lambda} \geq \frac{ (1-\log (2) ) L^4}{128 (n-1) ( d  v)^2   }.\label{eqn:tradeoff_rd}
\end{equation} 
\end{enumerate}

\begin{remark}\label{remark:corollary_1}
Notice that for both the random waypoint and random direction mobility models, if we consider that the   $L$ and $n$ increase while  the node density $\tau = n/L^2 $ remains constant, then we have the following observations:
\begin{itemize}
	\item The results of~(\ref{eqn:appr_capa}) reduce to $\mu_{\text{RW}} \approx  c_1 \tau d   \E[V^*]$ and $\mu_{\text{RW}} \approx   \tau d   \E[V^*]$, indicating that a constant throughput capacity is still achievable in  a large scale ICMN.
Meanwhile, the result in~(\ref{eqn:delay}) indicates that the average end-to-end delay under Algorithm~\ref{alg:routing} will increase linearly with the number of nodes $n$.
	\item The results in~(\ref{eqn:tradeoff_rw}) and~(\ref{eqn:tradeoff_rd})  indicate that the delay-throughput scales as ${\E\{D\}} / {\lambda} > O(n)$.
\end{itemize}

\end{remark}

\section{Simulation and Numerical Results}\label{sec:numerical}
In this section, we first provide simulation results to validate the efficiency of the theoretical results developed in Section~\ref{sec:capacity}, and then apply these results to illustrate the performance of the concerned ICMNs under different settings of system parameters.
\subsection{Model Validation}
To validate the efficiency of our analytical results, we provide simulation results under the random waypoint and  the random direction mobility models in this section.
The simulation results were obtained from a self-developed  discrete event simulator that implements the packet delivery process under Algorithm~\ref{alg:routing} and accepts  mobility traces generated by the NS-$2$ code of the random waypoint and random direction mobility models as input.
\subsubsection{Mobility Models}\label{sec:numerical:mobility}
The mobility models considered in the simulation are summarized as follows.
\begin{itemize}

\item  Random waypoint mobility model~\cite{Groenevelt2005}:
Under this model, initially network nodes are uniformly distributed in the network area and each node travels at a  travel speed randomly and uniformly selected in $(v_\text{min}, v_\text{max})$ with $v_\text{min} > 0$ towards a destination randomly and uniformly selected in the network area. 
After arriving at the destination, the node may pause for a random amount of time and then chooses a new destination and a new travel speed, independently of previous ones.
It is notable that the  locations of the nodes in steady-state under the random waypoint model are not uniformly distributed.
Particularly, it was reported in~\cite{Bettstetter2002Proc.WMAN} that the stationary distribution of the location of a node  is more concentrated near the center of the network region.

	\item  Random direction mobility model~\cite{Groenevelt2005}: Under this mobility model, initially network nodes are uniformly distributed in the network area and each node randomly selects a direction, a speed and a finite traveling time.
The node travels towards the direction at the given speed for the given duration of time. 
When the travel time duration has expired,  the node could pause for a random time, after which it selects a new set of direction, speed and time duration, independently of all previous ones.
When the node reaches a boundary, it is either reflected (i.e., it is bounced back to the network
area with the angle of $\theta$ or $\pi - \theta$) or the area wraps around so that it appears on the other side.
It was shown in~\cite{Nain2005Proc.IEEEINFOCOM} that the stationary distribution of locations is uniformly distributed for arbitrary distributions of direction, speed and travel time duration, irrespective of the boundaries being reflecting or wrapped around.

\end{itemize}
\subsubsection{Simulation Setting}
In our simulation, we consider a square network  of side-length $L = 2000$ m and  number of nodes $n = 20$.
The travel speed is constant and equals to $v = 40$ m$/$s.
There is no pause time.
We consider transmission distances of $ d = \{20,50,100\}$, where according to~(\ref{eqn:beta_estimate}) the corresponding pairwise meeting rates are determined as $\beta_\text{RW} = \{6.96 \times 10^{-4},1.74\times 10^{-3}, 3.48 \times 10^{-3} \}$ for the random waypoint mobility model and $\beta_\text{RD} = \{5.09 \times 10^{-4},1.27\times 10^{-3}, 2.55 \times 10^{-3} \}$ for the random direction mobility model. 
For the simulation measurements of the throughput and average end-to-end delay under Algorithm~\ref{alg:routing}, we focus on a specific traffic flow and measure its throughput and average packet delay over a long time period
of $1.0 \times 10^7$ seconds for each system load $\rho = \lambda / \mu$.

\subsubsection{Simulation Results}
\begin{figure}
	\centering
	\subfloat[Random waypoint model.]{\includegraphics[width=3.0in]{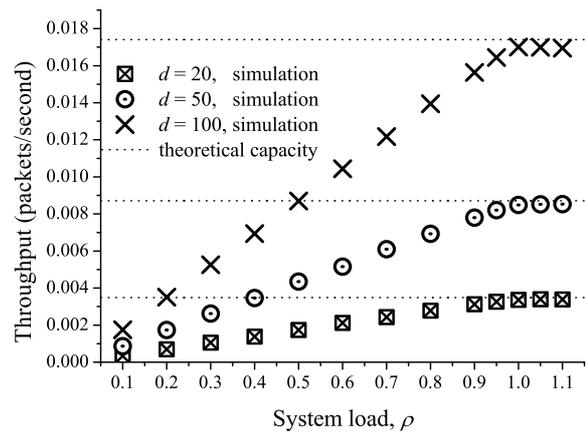}\label{fig:throughput_vs_load_RW}}\\
	\subfloat[Random direction model.]{\includegraphics[width=3.0in]{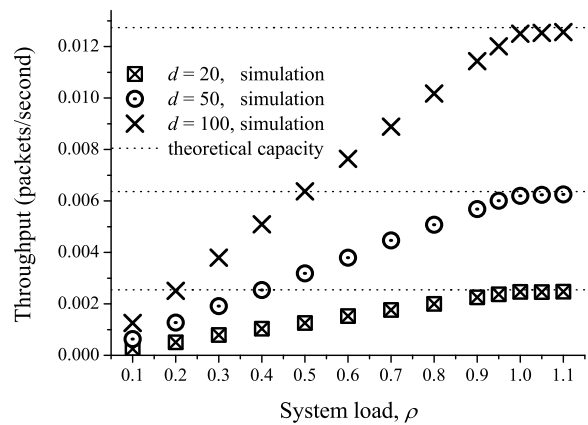}\label{fig:throughput_vs_load_RD}}
		\caption{Throughput vs. system load $\rho$.}\label{fig:throughput_vs_load}
\end{figure}
\begin{figure}
	\centering
		\subfloat[Random waypoint model.]{\includegraphics[width=3.0in]{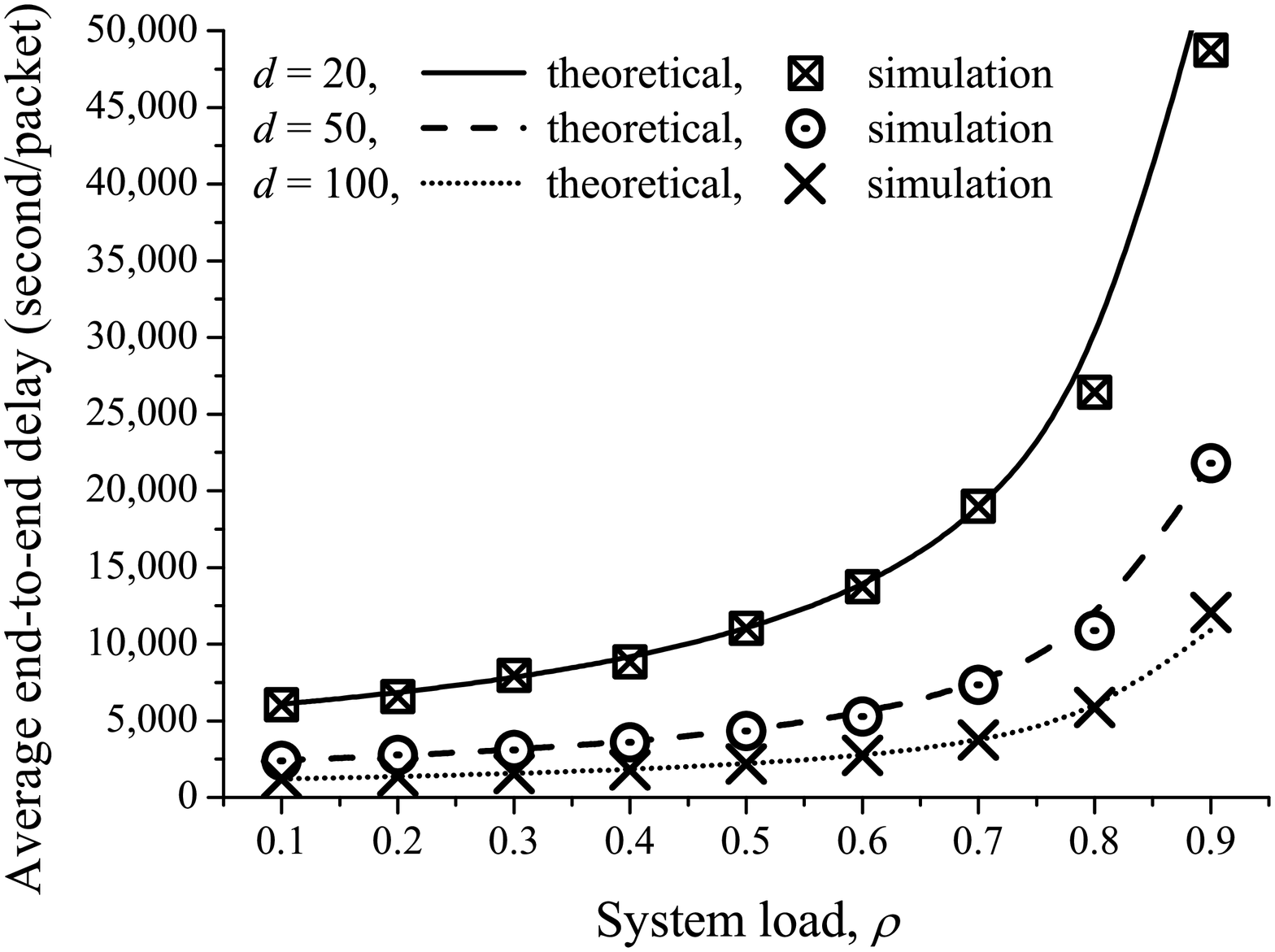}\label{fig:delay_vs_load_RW}}\\
		\subfloat[Random direction model.]{\includegraphics[width=3.0in]{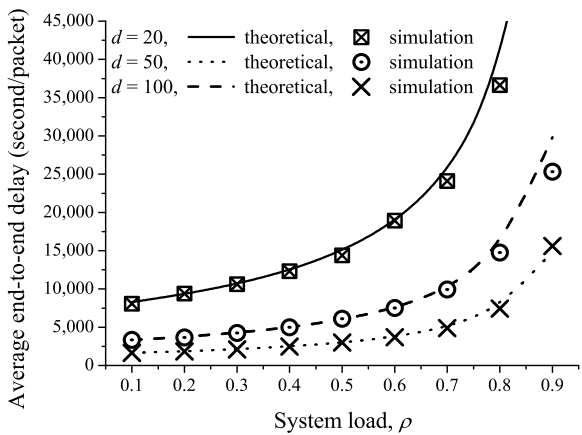}\label{fig:delay_vs_load_RD}}
		\caption{Average end-to-end delay vs. system load $\rho$.}\label{fig:delay_vs_load}
\end{figure}
To validate the efficiency of the  developed throughput capacity model, we summarize in Fig.~\ref{fig:throughput_vs_load} the simulation results of throughput for different values of system load.
In Fig.~\ref{fig:throughput_vs_load}, the dots represent  the  simulation results and the dashed lines are the corresponding theoretical throughput capacities calculated by~(\ref{eqn:appr_capa}).
We can observe from Fig.~\ref{fig:throughput_vs_load} that for both the random waypoint and random direction mobility models, the throughput  increases linearly as $\rho$ increases from $0$ to $1$ and approaches  $\mu$ when $\rho$ grows further beyond $1$.
This is expected since the queuing system in the network is underloaded when $\rho < 1$,
and it saturates as $\rho$ approaches $1$ and beyond.
The results in Fig.~\ref{fig:throughput_vs_load} indicate clearly that our theoretical throughput capacity result developed based on the Poisson meeting process can accurately predict the throughput capacity for the concerned ICMNs with the random waypoint or random direction mobility model.
Moreover, it also indicates that this throughput capacity can be achieved by adopting Algorithm~\ref{alg:routing} as routing algorithm in the network.

%

We then proceed to  validate the efficiency of our end-to-end delay model.
Particularly, we compare in Fig.~\ref{fig:delay_vs_load} the simulation results of the average end-to-end packet delay to those of theoretical ones calculated by substituting the results in~(\ref{eqn:appr_capa}) into~(\ref{eqn:delay}).
We can see from Fig.~\ref{fig:delay_vs_load} that for both the considered mobility models, the theoretical results nicely agree with the simulation ones. 
This observation indicates that our delay model of~(\ref{eqn:delay}) is accurate and can efficiently capture the delay behavior under Algorithm~\ref{alg:routing} in the considered network.

\subsection{Numerical Results and Discussions}

\begin{figure}
	\centering
		\includegraphics[width=3.0in]{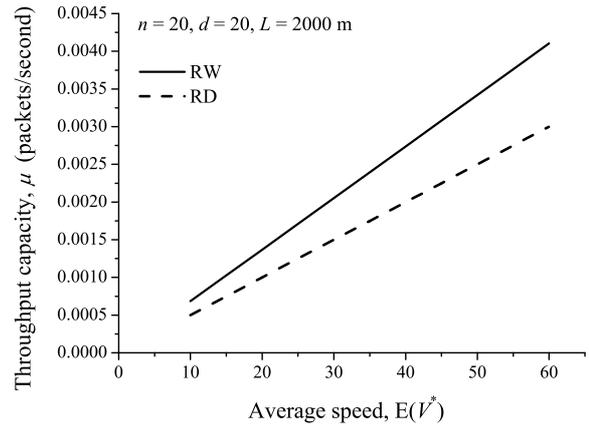}
		\caption{Capacity $\mu$ vs. average speed $\E\{V^*\}$.}
	\label{fig:mu_vs_speed}
\end{figure}
\begin{figure}
	\centering
		\includegraphics[width=3.0in]{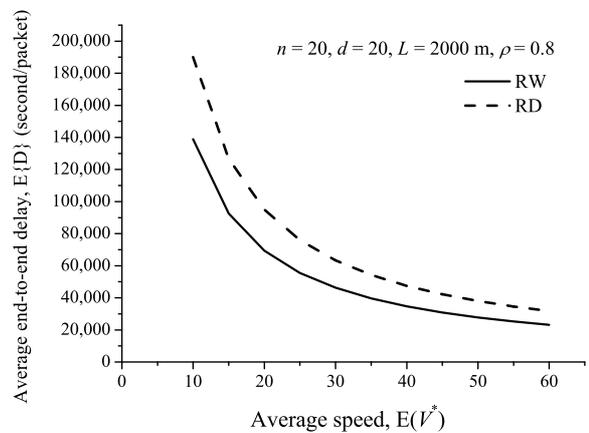}
		\caption{Average end-to-end delay $\E\{D\}$   vs. average speed $\E\{V^*\}$.}
	\label{fig:delay_vs_speed}
\end{figure}
\begin{figure}
	\centering
		\includegraphics[width=3.0in]{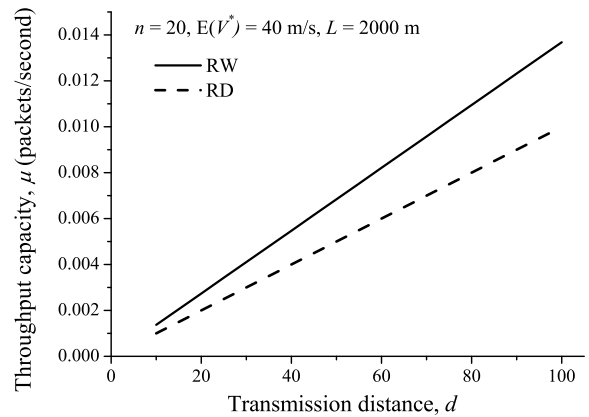}
		\caption{Capacity $\mu$ vs. transmission distance $d$.}
	\label{fig:mu_vs_d}
\end{figure}
\begin{figure}
	\centering
		\includegraphics[width=3.0in]{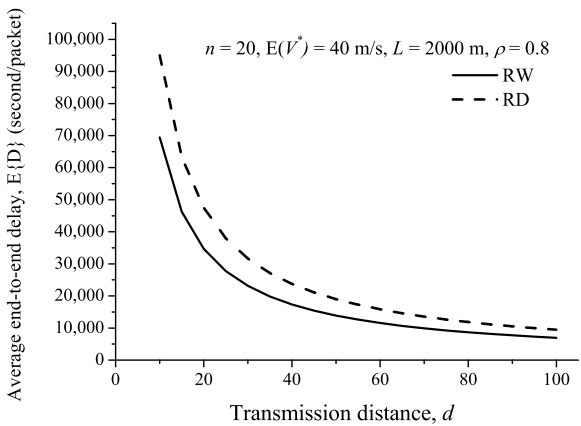}
		\caption{Average end-to-end delay $\E\{D\}$ vs. transmission distance $d$.}
	\label{fig:delay_vs_d}
\end{figure}

Based on our theoretical models, we first explore the impact of nodel traveling speed on the throughput capacity and end-to-end delay. 
We summarize in Fig.~\ref{fig:mu_vs_speed}  how the $\mu$ varies with average pairwise relative speed  $\E\{V^*\}$ in a network of $n = 20, d = 20$ m and $L=2000$ m.
Fig.~\ref{fig:mu_vs_speed}  shows that as the $\E\{V^*\}$ increases, the throughput capacities under  both the random waypoint and random direction models increase linearly.
This is mainly due to that a higher average travel speed will lead to an increase on the pairwise meeting rate as shown in~(\ref{eqn:beta_estimate}), and hence to a higher throughput capacity.
For the same network setting, we then present in Fig.~\ref{fig:delay_vs_speed} how the average delay $\E\{D\}$ under Algorithm~\ref{alg:routing} varies with $\E\{V^*\}$ under system load $\rho = 0.8$.
It can be observed in Fig.~\ref{fig:delay_vs_speed} that increasing $\E\{V^*\}$ will cause a lower average delay, which is because the $\E\{D\}$ is inverse proportional to the throughput capacity $\mu$ as indicated in~(\ref{eqn:delay}).

We then present in Fig.~\ref{fig:mu_vs_d} and~\ref{fig:delay_vs_d} how the throughput capacity $\mu$ and average end-to-end packet delay vary with transmission distance $d$ for a network of $n = 20, \E\{V^*\} = 40$ m/s, $L=2000$ m and $\rho = 0.8$ (for delay).
It can be seen from in Figs.~\ref{fig:mu_vs_d} and~\ref{fig:delay_vs_d} that the impacts  of the transmission distance $d$ on the behavior of capacity and delay are similar to those of the $\E\{V^*\}$, for the reason that as shown in~(\ref{eqn:beta_estimate}),  $d$ is also  a factor in the evaluation of $\beta$.

It is also interesting to see that from Figs.~\ref{fig:mu_vs_speed}-\ref{fig:delay_vs_d} that the random waypoint mobility model provides a performance better than that of the random direction mobility model for the network settings here.
Recall that compared with the random direction model that has a uniform stationary distribution of  nodes location, the stationary distribution of the location of a node under the random waypoint mobility model is more concentrated near the center of the network region (see Section~\ref{sec:numerical:mobility}).
Therefore,  the random waypoint mobility model leads to a higher nodel pairwise meeting rate (see~(\ref{eqn:beta_estimate})) and hence a higher throughput capacity, for the same network setting of $L$,  $\E\{V^*\}$ and $d$.

\section{Conclusions}\label{sec:conclusion}
This paper studied the  throughput capacity and delay-throughput tradeoff in an ICMN with Poisson meeting process. Based on the pairwise meeting rate in the concerned ICMN, an exact expression of the throughput capacity is derived, which indicates the maximum throughput that the network can stably support.
To reveal the inherent relationship between the end-to-end packet delay and achievable throughput, a necessary condition on the delay-throughput tradeoff is also established.
To illustrate the applicability of these theoretical results developed based on the Poisson meeting process, we conducted parameter-matching to fit the random waypoint and random direction models to the Poisson meeting process and obtained approximations to the throughput capacity and delay-throughput tradeoff with these mobility models.
Simulation result demonstrates that the throughput capacity developed based on the Poisson meeting process can serve as a good approximation to that under the random waypoint or random direction mobility models.
It is expected that the theoretical analysis developed in this paper will be also helpful for exploring the throughput capacity and delay-throughput tradeoff in ICMNs under other types of mobility models as well. 
Remark~\ref{remark:corollary_1} indicates that under the random waypoint or random direction mobility, a constant  throughput capacity is achievable even in a large scale ICMN as far as the node density can be kept constant, but at the cost of a linearly increasing expected end-to-end delay. Our results also reveal that by increasing the average node traveling speed or transmission range in an ICMN, an improvement on both its throughput and end-to-end delay performance might be expected.

\bibliographystyle{ieeetran}
\bibliography{IEEEabrv,ref}

\end{document}